\documentclass[conference]{IEEEtran}

\usepackage[cmex10]{amsmath}
\usepackage[dvipdfmx]{graphicx}
\usepackage{array}
\usepackage{mdwmath}
\usepackage{mdwtab}
\usepackage{amssymb}

\newtheorem{theorem}{Theorem}
\newtheorem{lemma}[theorem]{Lemma}

\newtheorem{corollary}[theorem]{Corollary}

\newcommand{\qed}{\hfill \IEEEQED}

\newcommand{\markov}{ - \!\!\circ\!\! - }
\DeclareMathAlphabet{\bm}{OML}{cmm}{b}{it}

\newcommand{\rom}[1]{\mathrm{#1}}
\newcommand{\san}[1]{\mathsf{#1}}

\begin{document}
%
\title{A Converse Bound on Wyner-Ahlswede-K\"orner Network via Gray-Wyner Network}

\author{\IEEEauthorblockN{Shun Watanabe\IEEEauthorrefmark{1}}
\IEEEauthorblockA{\IEEEauthorrefmark{1}Department of Computer and Information Sciences, 
Tokyo University of Agriculture and Technology, Japan, \\
E-mail:shunwata@cc.tuat.ac.jp}}


\maketitle

\begin{abstract}
We show a reduction method to construct a code for the Gray-Wyner (GW) network 
from a given code for the Wyner-Ahlswede-K\"orner (WAK) network. By combining this reduction with a 
converse bound on the GW network, we derive a converse bound on the WAK network. The derived bound 
gives an alternative proof of the strong converse theorem for the WAK network. 
\end{abstract}


%
\IEEEpeerreviewmaketitle

\section{Introduction}

We revisit the coding problem over the Wyner-Ahlswede-K\"orner network, which is 
also known as the lossless source coding with one-helper. The achievable rate region of 
this network was characterized in \cite{wyner:75c, ahlswede:75}. This network is regarded as 
one of typical problems of the network information theory
in the sense that it contains some of basic difficulties 
that arise in multiuser problems; in particular, the characterization of
the achievable rate region involves an auxiliary random variable and Markov chain structure,
which makes it difficult to derive converse bounds of this network. The strong converse theorem for
this network was proved by Ahlswede-G\'acs-K\"orner in \cite{ahlswede:76} with a technique called
the blowing-up lemma. The exponential strong converse was recently
shown by Oohama in \cite{Oohama:15} with some new techniques in the information-spectrum method.

The coding problem over the Gray-Wyner (GW) network
 is another basic problem of
the network information theory introduced in \cite{GraWyn:74}. The characterization of
the achievable rate region of this network also involves an auxiliary random variable; however, 
it does not involve Markov chain structure. The strong converse theorem for this network was
shown by Gu-Effros in \cite{GuEff:09}. By a type based refinement of their approach, the second-order
rate region of the GW network was shown in \cite{Watanabe:15}.

A motivation of this work is to develop an alternative converse approach to the WAK network.
Since the GW network is centralized coding while the WAK network is distributed coding,
it is not clear how to apply the approach in \cite{GuEff:09, Watanabe:15}
to the WAK network directly. 
However, we derive a converse bound on the WAK network by
showing a reduction from the GW network to the WAK network and then
by applying the converse approach  \cite{GuEff:09, Watanabe:15} of the GW network. 
In order to explain an overview of our approach, let us formally introduce
each network below.

\subsection{Gray-Wyner Network}

The coding system of the GW network consists of three encoders
\begin{align*}
\varphi_i^{(n)}:{\cal X}^n \times {\cal Y}^n \to {\cal M}_i^{(n)},~~~i=0,1,2
\end{align*}
and two decoders
\begin{align*}
& \psi_1^{(n)}:{\cal M}_0^{(n)} \times {\cal M}_1^{(n)} \to {\cal X}^n, \\
& \psi_2^{(n)}:{\cal M}_0^{(n)} \times {\cal M}_2^{(n)} \to {\cal Y}^n.
\end{align*}
We omit the blocklength $n$ when it is obvious from the context. For $(X^n,Y^n) \sim P$, the error probability $\rom{P}_{\mathtt{GW}}(\Phi_n|P)$
of code $\Phi_n = (\varphi_0,\varphi_1,\varphi_2,\psi_1,\psi_2)$ is defined as the probability such that 
$(\psi_i(\varphi_0(X^n,Y^n), \varphi_i(X^n,Y^n)) : i = 1,2) \neq (X^n,Y^n)$.
A rate triplet $(r_0,r_1,r_2)$ is defined to be achievable if there exists a sequence of code $\{ \Phi_n \}_{n=1}^\infty$ such that 
\begin{align*}
\limsup_{n\to \infty} \frac{1}{n} \log |{\cal M}_i^{(n)}| \le r_i,~~~ i=0,1,2 
\end{align*}
and
\begin{align*}
\lim_{n\to\infty} \rom{P}_{\mathtt{GW}}(\Phi_n | P_{XY}^n) = 0,
\end{align*}
where $P_{XY}^n$ is the product (i.i.d.) of $P_{XY}$. Then, the achievable rate region ${\cal R}_{\mathtt{GW}}(P_{XY})$ is defined as 
the set of all achievable rate triplets. 

Let ${\cal R}_{\mathtt{GW}}^*(P_{XY})$ be the set of all rate triplets $(r_0,r_1,r_2)$ such that there exists a test channel $P_{W|XY}$ with 
$|{\cal W}| \le |{\cal X}||{\cal Y}|+2$ satisfying 
\begin{align*}
r_0 \ge I(W \wedge X,Y), ~~~
r_1 \ge H(X|W), ~~~
r_2 \ge H(Y|W).
\end{align*} 
It is known that the achievable region of the GW network is characterized as 
${\cal R}_{\mathtt{GW}}(P_{XY}) = {\cal R}^*_{\mathtt{GW}}(P_{XY})$.

\subsection{Wyner-Ahlswede-K\"orner Network}

The coding system of the WAK network consists of two encoders\footnote{For later convenience of
relating the WAK network with the GW network, we use unconventional notations; the helper's encoder is $\tilde{\varphi}_0$
and the main encoder is $\tilde{\varphi}_2$.}
\begin{align*}
& \tilde{\varphi}_0^{(n)}: {\cal X}^n \to \tilde{{\cal M}}_0^{(n)}, \\
& \tilde{\varphi}_2^{(n)}: {\cal Y}^n \to \tilde{{\cal M}}_2^{(n)}
\end{align*}
and one decoder
\begin{align*}
\tilde{{\cal M}}_0^{(n)} \times \tilde{{\cal M}}_2^{(n)} \to {\cal Y}^n.
\end{align*}
For $(X^n,Y^n) \sim P$, the error probability $\rom{P}_{\mathtt{WAK}}(\tilde{\Phi}_n|P)$ of code $\tilde{\Phi}_n = (\tilde{\varphi}_0,\tilde{\varphi}_2,\tilde{\psi})$
is defined as the probability such that $\tilde{\psi}(\tilde{\varphi}_0(X^n),\tilde{\varphi}_2(Y^n)) \neq Y^n$.
A rate pair $(r_0,r_2)$ is defined to be achievable if there exists a sequence of code $\{ \tilde{\Phi}_n \}_{n=1}^\infty$ such that 
\begin{align*}
\limsup_{n\to\infty} \frac{1}{n} \log |\tilde{{\cal M}}_i^{(n)}| \le r_i,~~~i=0,2
\end{align*}
and
\begin{align*}
\lim_{n\to\infty} \rom{P}_{\mathtt{WAK}}(\tilde{\Phi}_n|P_{XY}^n) = 0.
\end{align*}
Then, the achievable region ${\cal R}_{\mathtt{WAK}}(P_{XY})$ is defined as the set of all achievable rate pairs. 

Let ${\cal R}_{\mathtt{WAK}}^*(P_{XY})$ be the set of all rate pair $(r_0,r_2)$ such that there exists a test channel $P_{W|X}$ with $|{\cal W}| \le |{\cal X}| |{\cal Y}|+2$
satisfying\footnote{In fact, we can show the cardinality bound $|{\cal W}| \le |{\cal X}| + 1$. However, for later convenience of
relating the WAK network with the GW network, we apply a slightly loose bound;
there is no harm in enlarging the cardinality of the auxiliary random variable.} 
\begin{align*}
r_0 \ge I(W \wedge X),~~~~ r_2 \ge H(Y|W).
\end{align*}
It is known that the achievable region of the WAK network is characterized as 
${\cal R}_{\mathtt{WAK}}(P_{XY}) = {\cal R}^*_{\mathtt{WAK}}(P_{XY})$.

\subsection{Overview of Approach}

Although the GW network and the WAK network appear to be completely different problems (the former is centralized encoding
while the latter is distributed encoding),
it is known that the achievable rate regions of these networks have the following intimate connection \cite{nair:16}:
\begin{align}
\lefteqn{ \{ (r_0,r_1,r_2) \in {\cal R}_{\mathtt{GW}}^*(P_{XY}) : r_0 + r_1 = H(X) \} } \nonumber \\
&= \{ (r_0,r_1,r_2) : (r_0,r_2) \in {\cal R}_{\mathtt{WAK}}^*(P_{XY}), r_0 + r_1 = H(X) \}. \label{eq:connection-WAK-GW}
\end{align}
In fact, by noting the identity $H(X) + I(W\wedge Y|X) = I(W\wedge X,Y)+ H(X|W)$, 
we can verify that the condition $I(W\wedge X,Y) + H(X|W) = H(X)$ enforces the Markov chain condition 
$W \markov X \markov Y$.

Inspired by the connection in \eqref{eq:connection-WAK-GW},
we shall show a converse bound on the WAK network by the following reduction argument. 
For a given WAK code $\tilde{\Phi}_n = (\tilde{\varphi}_0,\tilde{\varphi}_2,\tilde{\psi})$ with rates $(\tilde{r}_0,\tilde{r}_2)$, we construct a GW code
$\Phi_n = (\varphi_0,\varphi_1,\varphi_2,\psi_1,\psi_2)$ with rates $(r_0,r_1,r_2)$ such that $r_0 \simeq \tilde{r}_0$,
$r_2 \simeq \tilde{r}_2$, and $r_1 \simeq H(X) - r_0$; we also show that  the error probability of the constructed GW code $\Phi_n$ is as small as
that of the original WAK code $\tilde{\Phi}_n$. Then, we apply a converse bound on the GW network, which gives a converse bound on
the WAK network via the above reduction argument. Our approach gives an alternative proof of the
strong converse for the WAK network without using the blowing-up lemma nor Oohama's method.

The rest of the paper is organized as follows. 
In the next section, we state our main results. All the proofs are given in
Section \ref{sec:proofs}. We close the paper with some discussions in Section \ref{section:discussion}. 



\section{Main Result} \label{sec:main}

For a joint type $P_{\bar{X}\bar{Y}} \in {\cal P}_n({\cal X}\times {\cal Y})$,
let $P_{{\cal T}^n_{\bar{X}\bar{Y}}}$ be the uniform distribution on the joint type class ${\cal T}^n_{\bar{X}\bar{Y}}$.
The main result of this paper is the following reduction theorem claiming that 
we can construct a GW code from a given WAK code. 
\begin{theorem} \label{theorem:operational-connection}
For a given WAK code $\tilde{\Phi}_n = (\tilde{\varphi}_0,\tilde{\varphi}_2,\tilde{\psi})$ and a joint type $P_{{\cal T}^n_{\bar{X}\bar{Y}}}$ satisfying 
\begin{align}
\log |{\cal T}^n_{\bar{X}}| \ge \log |\tilde{{\cal M}}_0|, 
 \label{eq:condition-M0-X}
\end{align}
where ${\cal T}^n_{\bar{X}}$ is the type class of the marginal type $P_{\bar{X}}$, there exists a GW code $\Phi_n = (\varphi_0,\varphi_1,\varphi_2,\psi_1,\psi_2)$ such that 
\begin{align}
\log |{\cal M}_0| &\le \log |\tilde{{\cal M}}_0| + \log n + \log \log |{\cal X}| + 2, \label{eq:connection-1} \\
\log |{\cal M}_0| |{\cal M}_1| &\le \log |{\cal T}^n_{\bar{X}}| + \log n + \log \log |{\cal X}| + 2, \label{eq:connection-2} \\
\log |{\cal M}_2| &= \log |\tilde{{\cal M}}_2|, \label{eq:connection-3}
\end{align}
and 
\begin{align}
\rom{P}_{\mathtt{GW}}(\Phi_n | P_{{\cal T}^n_{\bar{X}\bar{Y}}}) \le \rom{P}_{\mathtt{WAK}}(\tilde{\Phi}_n | P_{{\cal T}^n_{\bar{X}\bar{Y}}}).
 \label{eq:connection-4}
\end{align}
\end{theorem}

Next, we shall derive a converse bound on the WAK network by combining Theorem \ref{theorem:operational-connection} with a converse bound on the GW network. 
For that purpose, let us introduce a slightly relaxed version ${\cal R}^*_{\mathtt{WAK}}(\delta|P_{XY})$
of ${\cal R}^*_{\mathtt{WAK}}(P_{XY})$ as follows. For $\delta > 0$, let ${\cal R}^*_{\mathtt{WAK}}(\delta|P_{XY})$ be
the set of all rate pairs $(r_0,r_2)$ such that there exists a test channel $P_{W|XY}$ with
$|{\cal W}| \le |{\cal X}||{\cal Y}|+2$ satisfying 
\begin{align*}
r_0 \ge I(W \wedge X,Y),~~~ r_2 \ge H(Y|W),~~~ \delta \ge I(W \wedge Y | X).
\end{align*}
Note that ${\cal R}^*_{\mathtt{WAK}}(0|P_{XY}) = {\cal R}^*_{\mathtt{WAK}}(P_{XY})$.
\begin{corollary} \label{corollary:WAK-bound}
For a given WAK code $\tilde{\Phi}_n = (\tilde{\varphi}_0,\tilde{\varphi}_2,\tilde{\psi})$, it hold that 
\begin{align*}
\lefteqn{ \rom{P}_{\mathtt{WAK}}(\tilde{\Phi}_n | P^n_{XY})  } \\
& \ge \rom{P}\bigg( (\tilde{r}_{0,n}, \tilde{r}_{2,n}) \notin {\cal R}^*_{\mathtt{WAK}}(\delta_n | \san{t}_{X^nY^n}),~\san{t}_{X^n} \in {\cal E}_n \bigg)\bigg(1-\frac{1}{n} \bigg),
\end{align*}
where $\san{t}_{X^n Y^n}$ is the joint type of $(X^n,Y^n)$, 
\begin{align*}
\tilde{r}_{0,n} &:= \frac{1}{n} \log |\tilde{{\cal M}}_0| + \Delta_n + \frac{\log\log |{\cal X}| +2}{n}, \\
\tilde{r}_{2,n} &:= \frac{1}{n} \log |\tilde{{\cal M}}_2| + \frac{1 + \log |{\cal Y}|}{n}, \\
\delta_n &:= \Delta_n + \frac{\log\log |{\cal X}| + 3 + \log |{\cal X}|}{n},
\end{align*}
$\Delta_n = \frac{(|{\cal X}|(|{\cal Y}|+1)+3)\log(n+1)}{n}$, and ${\cal E}_n$ is the set of types defined by
\begin{align*}
{\cal E}_n := \bigg\{ P_{\bar{X}}  : H(\bar{X}) \ge \frac{1}{n} \log |\tilde{{\cal M}}_0| + \frac{|{\cal X}|\log(n+1)}{n} \bigg\}.
\end{align*}
\end{corollary}

By noting the continuity of region ${\cal R}^*_{\mathtt{WAK}}(\delta|P_{XY})$ at $\delta = 0$, we can show the following
strong converse theorem for the WAK network. 
\begin{corollary}[\cite{ahlswede:76}] \label{corollary:strong-converse}
If $(r_0, r_2) \notin {\cal R}_{\mathtt{WAK}}^*(P_{XY})$ and $r_0 < H(X)$,\footnote{Our approach only gives the strong converse 
theorem under the condition $r_0 < H(X)$; however, the strong converse theorem for the WAK network is known to hold without this
condition \cite{ahlswede:76}. In fact, for $r_0 \ge H(X)$, it can be shown as the strong converse theorem for
the Slepian-Wolf network with full side-information.} 
then for any sequence of WAK codes $\{ \tilde{\Phi}_n \}_{n=1}^\infty$ satisfying
\begin{align}
\limsup_{n\to\infty} \frac{1}{n} \log |\tilde{{\cal M}}^{(n)}_0| &\le r_0, \label{eq:rate-condition-0} \\
\limsup_{n\to\infty} \frac{1}{n} \log |\tilde{{\cal M}}^{(n)}_2| &\le r_2, \label{eq:rate-condition-2}
\end{align}
it holds that 
\begin{align}
\lim_{n\to\infty} \rom{P}_{\mathtt{WAK}}(\tilde{\Phi}_n | P_{XY}^n) = 1. \label{eq:convergence-to-one}
\end{align}
\end{corollary}

\section{Proofs} \label{sec:proofs}

\subsection{Proof of Theorem \ref{theorem:operational-connection}}

For a given WAK code $\tilde{\Phi}_n=(\tilde{\varphi}_0,\tilde{\varphi}_2,\tilde{\psi})$,
the encoder $\tilde{\varphi}_0$ induces a partition $\tilde{\varphi}_0^{-1}(m) \cap {\cal T}_{\bar{X}}^n$,
$m \in \tilde{{\cal M}}_0$ of the type class ${\cal T}_{\bar{X}}^n$. 
Basic strategy to construct encoder $\varphi_1$ is to assign distinct codewords to each element in $\tilde{\varphi}_0^{-1}(m) \cap {\cal T}_{\bar{X}}^n$;
however, some partitions may have much larger cardinality than others. 
The following lemma states that, with a negligible penalty rate, we can construct a modified WAK code having  ``balanced" property,  
from which Theorem \ref{theorem:operational-connection} follows immediately. 
\begin{lemma}[Balanced Code] \label{lemma:balanced}
For a given WAK code $\tilde{\Phi}_n=(\tilde{\varphi}_0,\tilde{\varphi}_2,\tilde{\psi})$ and a joint type $P_{\bar{X}\bar{Y}}$ 
satisfying \eqref{eq:condition-M0-X}, there exists another WAK code $\hat{\Phi}_n = (\hat{\varphi}_0,\hat{\varphi}_2,\hat{\psi})$ such that 
\begin{align}
\log |\hat{{\cal M}}_0| &\le \log |\tilde{{\cal M}}_0| + \log n + \log\log |{\cal X}| + 2, \label{eq:balanced-1} \\
\log |\hat{{\cal M}}_2| &= \log |\tilde{{\cal M}}_2|, \label{eq:balanced-2}  \\
\rom{P}_{\mathtt{WAK}}(\hat{\Phi}_n| P_{{\cal T}_{\bar{X}\bar{Y}}^n}) &\le \rom{P}_{\mathtt{WAK}}(\tilde{\Phi}_n| P_{{\cal T}_{\bar{X}\bar{Y}}^n}), \label{eq:balanced-3}
\end{align}
and 
\begin{align}
\log |\hat{\varphi}_0^{-1}(m) \cap {\cal T}_{\bar{X}}^n| \le \log \frac{|{\cal T}_{\bar{X}}^n|}{|\tilde{{\cal M}}_0|} \label{eq:balanced-4}
\end{align}
for every $m \in \hat{{\cal M}}_0$.
\end{lemma}
\begin{proof}
Let\footnote{For simplicity, we assume $\log |\tilde{{\cal M}}_0^{(n)}| $ is an integer. } 
\begin{align*}
L_n := \log |\tilde{{\cal M}}_0|  
\le n \log |{\cal X}| .
\end{align*}
Let\footnote{This step is inspired by the information-spectrum slicing \cite{han:book}.} 
\begin{align*}
\tilde{{\cal M}}_0 = \bigcup_{i=0}^{L_n} \tilde{{\cal M}}_0(i)
\end{align*}
be the partition of $\tilde{{\cal M}}_0$, where 
\begin{align*}
\tilde{{\cal M}}_0(i) = \bigg\{ m  : \frac{|{\cal T}_{\bar{X}}^n|}{|\tilde{{\cal M}}_0|} 2^{(i-1)} < 
  |\tilde{\varphi}_0^{-1}(m) \cap {\cal T}_{\bar{X}}^n |  \le  \frac{|{\cal T}_{\bar{X}}^n|}{|\tilde{{\cal M}}_0|} 2^i \bigg\}
\end{align*}
for $1 \le i \le L_n$ and 
\begin{align*}
\tilde{{\cal M}}_0(0) = \bigg\{ m : |\tilde{\varphi}_0^{-1}(m) \cap {\cal T}_{\bar{X}}^n |  \le \frac{|{\cal T}_{\bar{X}}^n|}{|\tilde{{\cal M}}_0^{(n)}|} \bigg\}.
\end{align*}
Then, for $1 \le i \le L_n$, we have
\begin{align} \label{eq:bound-for-ith-slice}
|\tilde{{\cal M}}_0(i) | \le \frac{|\tilde{{\cal M}}_0|}{2^{(i-1)}};
\end{align}
otherwise, we have
\begin{align*}
\left| \bigcup_{m \in \tilde{{\cal M}}_0(i)} \tilde{\varphi}_0^{-1}(m) \cap {\cal T}_{\bar{X}}^n \right| > |{\cal T}_{\bar{X}}^n|,
\end{align*}
which is a contradiction. To construct $\hat{\varphi}_0$, for each $1 \le i \le L_n$ and $m \in \tilde{{\cal M}}_0(i)$,
we further partition $\tilde{\varphi}_0^{-1}(m)$ into $2^i$ subsets so that 
\begin{align*}
|\hat{\varphi}_0^{-1}(\hat{m}) \cap {\cal T}_{\bar{X}}^n | &\le \frac{|{\cal T}_{\bar{X}}^n|}{|\tilde{{\cal M}}_0|}
\end{align*}
for every $\hat{m} \in \hat{{\cal M}}_0(i)$, where $\hat{{\cal M}}_0(i)$ is the set of indices induced by such a partition. 
Then, we have
\begin{align*}
| \hat{{\cal M}}_0(i) | = 2^i |\tilde{{\cal M}}_0(i) | \le 2 |\tilde{{\cal M}}_0|,
\end{align*}
where the last inequality follows from \eqref{eq:bound-for-ith-slice}.
For $m \in \tilde{{\cal M}}_0(0)$, 
we keep $\hat{\varphi}_0^{-1}(m)$ unchanged from $\tilde{\varphi}_0^{-1}(m)$, and thus $\hat{{\cal M}}_0(0) = \tilde{{\cal M}}_0(0)$.
On the other hand, we set $\hat{\varphi}_2 = \tilde{\varphi}_2$. 
By noting
\begin{align*}
|\hat{{\cal M}}_0| = \sum_{i=0}^{L_n} |\hat{{\cal M}}_0(i)| 
\le (2L_n + 1) |\tilde{{\cal M}}_0| 
\le 4 L_n |\tilde{{\cal M}}_0|,
\end{align*}
we can verify that the encoders constructed in
this manner satisfy \eqref{eq:balanced-1}, \eqref{eq:balanced-2}, and \eqref{eq:balanced-4}. 
Furthermore, since $\hat{\varphi}_0$ is finer than $\tilde{\varphi}_0$, we can construct a decoder $\hat{\psi}$ 
satisfying \eqref{eq:balanced-3}.
\end{proof}

Now, we prove Theorem \ref{theorem:operational-connection}.
For a given WAK code $\tilde{\Phi}_n = (\tilde{\varphi}_0,\tilde{\varphi}_2,\tilde{\psi})$, by Lemma \ref{lemma:balanced},
we can construct a WAK code $\hat{\Phi}_n = (\hat{\varphi}_0,\hat{\varphi}_2,\hat{\psi})$
satisfying \eqref{eq:balanced-1}-\eqref{eq:balanced-4}. We set $\varphi_0 = \hat{\varphi}_0$ and $\varphi_2 = \hat{\varphi}_2$. 
We take ${\cal M}_1$ so that 
\begin{align*}
|{\cal M}_1| = \max_{ m \in {\cal M}_0} | \varphi_0^{-1}(m) \cap {\cal T}_{\bar{X}}^n |,
\end{align*}
and we construct $\varphi_1$ so that distinct numbers are assigned to the elements
in $\varphi_0^{-1}(m) \cap {\cal T}_{\bar{X}}^n$ for each $m \in {\cal M}_0$.
By \eqref{eq:balanced-1}, \eqref{eq:balanced-2}, and \eqref{eq:balanced-4}, the encoders 
constructed in this manner satisfy \eqref{eq:connection-1}-\eqref{eq:connection-3}. 
Furthermore, since $(\varphi_0(\bm{x}),\varphi_1(\bm{x})) \neq (\varphi_0(\bm{x}^\prime),\varphi_1(\bm{x}^\prime))$
for any $\bm{x} \neq \bm{x}^\prime \in {\cal T}_{\bar{X}}^n$, there exists a decoder $\psi_1$ that can 
reconstruct $X^n$ without an error under the distribution $P_{{\cal T}_{\bar{X}\bar{Y}}^n}$. 
Thus, by using $\hat{\psi}$ for $\psi_2$, \eqref{eq:connection-4} is also satisfied. \qed

\subsection{Proof of Corollary \ref{corollary:WAK-bound}}

To prove Corollary \ref{corollary:WAK-bound}, we combine Theorem \ref{theorem:operational-connection}
with the following converse bound on the GW network, which is a type based refinement of the 
strong converse of the GW network derived in \cite{GuEff:09}.
\begin{lemma}(\cite[Lemma 6]{Watanabe:15}) \label{lemma:converse-GW}
For a given GW code $\Phi_n$, suppose that the probability of error satisfies 
\begin{align*}
1 - \rom{P}_{\mathtt{GW}}(\Phi_n | P_{{\cal T}_{\bar{X}\bar{Y}}^n}) \ge 2^{-n\alpha_n}
\end{align*}
for some positive $\alpha_n$. Let $\beta_n$ be another positive number. Then there exists
$P_{\bar{W}|\bar{X}\bar{Y}}$ with $|{\cal W}| \le |{\cal X}||{\cal Y}|+2$ such that 
\begin{align*}
\frac{1}{n} \log |{\cal M}_0^{(n)}| &\ge I(\bar{W} \wedge \bar{X},\bar{Y}) \\ 
 &~~ - \frac{|{\cal X}||{\cal Y}| \log (n+1)}{n} - (\alpha_n + \beta_n), \\
\frac{1}{n} \log |{\cal M}_1^{(n)}| &\ge H(\bar{X}|\bar{W}) - \frac{1}{n} - 2^{-n\beta_n} \log |{\cal X}|, \\
\frac{1}{n} \log |{\cal M}_2^{(n)}| &\ge H(\bar{Y}|\bar{W}) - \frac{1}{n} - 2^{-n\beta_n} \log |{\cal Y}|,
\end{align*}
where $(\bar{X},\bar{Y}) \sim P_{\bar{X}\bar{Y}}$.
\end{lemma}

To prove Corollary \ref{corollary:WAK-bound}, we first decompose the error probability by type ${\cal P}_n({\cal X}\times {\cal Y})$ as 
\begin{align*}
\rom{P}_{\mathtt{WAK}}(\tilde{\Phi}_n | P_{XY}^n) 
&= \sum_{P_{\bar{X}\bar{Y}} \in {\cal P}_n({\cal X}\times{\cal Y})} P_{XY}^n({\cal T}_{\bar{X}\bar{Y}}^n) \rom{P}_{\mathtt{WAK}}(\tilde{\Phi}_n | P_{{\cal T}_{\bar{X}\bar{Y}}^n}) \\
&\ge \sum_{P_{\bar{X}\bar{Y}} \in {\cal P}_n({\cal X}\times{\cal Y}): \atop P_{\bar{X}} \in {\cal E}_n} P_{XY}^n({\cal T}_{\bar{X}\bar{Y}}^n) \rom{P}_{\mathtt{WAK}}(\tilde{\Phi}_n | P_{{\cal T}_{\bar{X}\bar{Y}}^n}).
\end{align*}
For each joint type $P_{\bar{X}\bar{Y}}$ satisfying $P_{\bar{X}} \in {\cal E}_n$, 
there exists (possibly different codes for different joint types)
a GW code $\Phi_n = (\varphi_0,\varphi_1,\varphi_2,\psi_1,\psi_2)$ satisfying \eqref{eq:connection-1}-\eqref{eq:connection-4} 
of Theorem \ref{theorem:operational-connection}. By Lemma \ref{lemma:converse-GW} with $\alpha_n=\beta_n = \frac{\log n}{n}$, if 
\begin{align*}
r_{0,n} &:= \frac{1}{n} \log |{\cal M}_0^{(n)}| + \frac{|{\cal X}||{\cal Y}| \log (n+1)}{n} + (\alpha_n + \beta_n), \\
r_{1,n} &:= \frac{1}{n} \log |{\cal M}_1^{(n)}| + \frac{1}{n} + 2^{- n \beta_n} \log |{\cal X}|, \\
r_{2,n} &:= \frac{1}{n} \log |{\cal M}_2^{(n)}| + \frac{1}{n} + 2^{- n \beta_n} \log |{\cal Y}|
\end{align*} 
are such that $(r_{0,n},r_{1,n},r_{2,n}) \notin {\cal R}^*_{\mathtt{GW}}(P_{\bar{X}\bar{Y}})$, then 
\begin{align} \label{eq:GW-lower-bound-for-outside}
\rom{P}_{\mathtt{GW}}( \Phi_n | P_{{\cal T}_{\bar{X}\bar{Y}}^n}) > 1 - 2^{-n\alpha_n}.
\end{align}

We claim that $(r_{0,n},r_{1,n},r_{2,n}) \in {\cal R}^*_{\mathtt{GW}}(P_{\bar{X}\bar{Y}})$ implies $(\tilde{r}_{0,n}, \tilde{r}_{2,n}) \in {\cal R}^*_{\mathtt{WAK}}(\delta_n | P_{\bar{X}\bar{Y}})$.
In fact, when $(r_{0,n},r_{1,n},r_{2,n}) \in {\cal R}^*_{\mathtt{GW}}(P_{\bar{X}\bar{Y}})$, then there exists $P_{\bar{W}|\bar{X}\bar{Y}}$ such that 
\begin{align}
r_{0,n} &\ge I(\bar{W} \wedge \bar{X},\bar{Y}), \label{eq:bound-r-0} \\
r_{1,n} &\ge H(\bar{X}|\bar{W}), \label{eq:bound-r-1} \\
r_{2,n} &\ge H(\bar{Y}|\bar{W}). \label{eq:bound-r-2}
\end{align}
From \eqref{eq:connection-1} and \eqref{eq:bound-r-0}, we have 
\begin{align*}
\tilde{r}_{0,n} \ge I(\bar{W} \wedge \bar{X}, \bar{Y}).
\end{align*}
From \eqref{eq:connection-3} and \eqref{eq:bound-r-2}, we have 
\begin{align*}
\tilde{r}_{2,n} \ge H(\bar{Y}|\bar{W}).
\end{align*}
From \eqref{eq:connection-2}, \eqref{eq:bound-r-0}, and \eqref{eq:bound-r-1}, we have
\begin{align*}
\lefteqn{ H(\bar{X}) + I(\bar{W} \wedge \bar{Y}|\bar{X}) } \\
&= I(\bar{W} \wedge \bar{X},\bar{Y}) + H(\bar{X}|\bar{W}) \\
&\le r_{0,n} + r_{1,n} \\
&\le \frac{1}{n} \log |{\cal T}_{\bar{X}}^n| \\
&~~~+ \frac{3\log n + \log \log |{\cal X}| + \log |{\cal X}| + 3 + |{\cal X}||{\cal Y}| \log (n+1) }{n}  \\
&\le H(\bar{X}) + \delta_n.
\end{align*}
Thus, we have $(\tilde{r}_{0,n}, \tilde{r}_{2,n}) \in {\cal R}^*_{\mathtt{WAK}}(\delta_n | P_{\bar{X}\bar{Y}})$.
By taking the contraposition and by \eqref{eq:GW-lower-bound-for-outside}, 
if $(\tilde{r}_{0,n}, \tilde{r}_{2,n}) \notin {\cal R}^*_{\mathtt{WAK}}(\delta_n | P_{\bar{X}\bar{Y}})$, then we have
\begin{align*}
\rom{P}_{\mathtt{WAK}}(\tilde{\Phi}_n | P_{{\cal T}_{\bar{X}\bar{Y}}^n}) \ge \rom{P}_{\mathtt{GW}}(\Phi_n | P_{{\cal T}_{\bar{X}\bar{Y}}^n}) > 1 - 2^{- n \alpha_n}.
\end{align*}
Thus, we have
\begin{align*}
\lefteqn{ \rom{P}_{\mathtt{WAK}}(\tilde{\Phi}_n | P_{XY}^n) } \\
&\ge \sum_{P_{\bar{X}\bar{Y}} \in {\cal P}_n({\cal X}\times{\cal Y}): \atop P_{\bar{X}} \in {\cal E}_n} P_{XY}^n({\cal T}_{\bar{X}\bar{Y}}^n) \rom{P}_{\mathtt{WAK}}(\tilde{\Phi}_n | P_{{\cal T}_{\bar{X}\bar{Y}}^n}) \\
&\ge \sum_{P_{\bar{X}\bar{Y}} \in {\cal P}_n({\cal X}\times{\cal Y}) : \atop (\tilde{r}_{0,n},\tilde{r}_{2,n}) \notin {\cal R}^*_{\mathtt{WAK}}(\delta_n | P_{\bar{X}\bar{Y}}) , P_{\bar{X}} \in {\cal E}_n }
 P_{XY}^n({\cal T}_{\bar{X}\bar{Y}}^n) (1 - 2^{-n \alpha_n}) \\
 &= \rom{P}\bigg( (\tilde{r}_{0,n}, \tilde{r}_{2,n}) \notin {\cal R}^*_{\mathtt{WAK}}(\delta_n | \san{t}_{X^n Y^n}),~\san{t}_{X^n} \in {\cal E}_n \bigg) \bigg(1- \frac{1}{n} \bigg).
\end{align*}
\qed

\subsection{Proof of Corollary \ref{corollary:strong-converse}}

To discuss the continuity of region ${\cal R}^*_{\mathtt{WAK}}(\delta | P_{XY})$ at $\delta=0$, let us consider 
the following supporting line of the region:
\begin{align*}
R_\mu(\delta | P_{XY}) := \min \{ r_0 + \mu r_2 : (r_0,r_2) \in {\cal R}_{\mathtt{WAK}}^*(\delta | P_{XY}) \}
\end{align*}
for $\mu \ge 0$. For brevity, we write $R_\mu(P_{XY}) = R_\mu(0|P_{XY})$.
\begin{lemma} \label{lemma:continuity-1}
For a given $P_{XY}$ and $\mu \ge 0$, we have
\begin{align*}
\lim_{\delta \to 0} R_\mu(\delta | P_{XY}) = R_\mu(P_{XY}).
\end{align*}
\end{lemma}
\begin{proof}
By definition, $R_\mu(\delta |P_{XY}) \le R_\mu(P_{XY})$ for any $\delta > 0$. 
Let $P_{W|XY}$ be a test channel such that 
\begin{align*}
I(W \wedge X,Y) + \mu H(Y|W) = R_\mu(\delta |P_{XY})
\end{align*}
and $I(W \wedge Y|X) \le \delta$.
Let $P_{\tilde{W}\tilde{X}\tilde{Y}} = P_{W|X} P_{XY}$. Note that $P_{\tilde{X}\tilde{Y}} = P_{XY}$ and 
$\tilde{W} \markov \tilde{X} \markov \tilde{Y}$. By noting that 
$D(P_{WXY} \| P_{\tilde{W}\tilde{X}\tilde{Y}}) = I(W \wedge Y|X) \le \delta$ and by the
Pinsker inequality, we have
$\| P_{WXY} - P_{\tilde{W}\tilde{X}\tilde{Y}} \|_1 \le \sqrt{\delta/2}$.
Thus, by the continuity of the entropy, there exists $\delta^\prime$ such that 
$\delta^\prime \to 0$ as $\delta \to 0$ and 
\begin{align*}
R_\mu(P_{XY}) &\le I(\tilde{W} \wedge \tilde{X}) + \mu H(\tilde{Y}|\tilde{W}) \\
&= I(\tilde{W} \wedge \tilde{X},\tilde{Y}) + \mu H(\tilde{Y}|\tilde{W}) \\
&\le I(W \wedge X,Y) + \mu H(Y|W) + \delta^\prime \\
&= R_\mu(\delta | P_{XY}) + \delta^\prime,
\end{align*}
which implies the claim of the lemma.
\end{proof}

We also have the following continuity.
\begin{lemma} \label{lemma:continuity-2}
$R_\mu(P_{XY})$ is continuous with respect to $P_{XY}$. 
\end{lemma}
\begin{proof}
Let $P_{XY}$ and $P_{\tilde{X}\tilde{Y}}$ be such that $\| P_{XY} - P_{\tilde{X}\tilde{Y}} \|_1 \le \epsilon$. 
Let $P_{W|X}$ be a test channel such that 
\begin{align*}
I(W \wedge X,Y) + \mu H(Y|W) = R_\mu(P_{XY}).
\end{align*}
Let $P_{\tilde{W}\tilde{X}\tilde{Y}} = P_{W|X} P_{\tilde{X}\tilde{Y}}$. Then, 
\begin{align*}
\| P_{\tilde{W}\tilde{X}\tilde{Y}} - P_{WXY} \|_1 = \| P_{\tilde{X}\tilde{Y}} - P_{XY} \|_1 \le \epsilon.
\end{align*}
Thus, by the continuity of the entropy, there exists $\epsilon^\prime$ such that $\epsilon^\prime \to 0$ as
$\epsilon \to 0$ and 
\begin{align*}
R_\mu(P_{\tilde{X}\tilde{Y}})
&\le I(\tilde{W} \wedge \tilde{X}) + \mu H(\tilde{Y} | \tilde{W}) \\
&\le R_\mu(P_{XY}) + \epsilon^\prime. 
\end{align*}
Similarly, we can show $R_\mu(P_{XY}) \le R_\mu(P_{\tilde{X}\tilde{Y}}) + \epsilon^\prime$.
\end{proof}

Now, we prove Corollary \ref{corollary:strong-converse} by using Corollary \ref{corollary:WAK-bound}.
In the following, we use the same notations $(\tilde{r}_{0,n},\tilde{r}_{2,n},\delta_n)$ as Corollary \ref{corollary:WAK-bound}.

Let ${\cal K}_n \subseteq {\cal P}_n({\cal X}\times{\cal Y})$ be the set of all joint types $P_{\bar{X}\bar{Y}}$ such that
\begin{align*}
| P_{\bar{X}\bar{Y}}(x,y) - P_{XY}(x,y) | \le \sqrt{\frac{\log n}{n}}
\end{align*}
for every $(x,y) \in {\cal X}\times {\cal Y}$. By the Hoeffding inequality, we have
\begin{align*}
\rom{P}\bigg( \san{t}_{X^n Y^n} \in {\cal K}_n \bigg) \ge 1 - \frac{2 |{\cal X}||{\cal Y}|}{n^2}.
\end{align*}
Since $r_0 < H(X)$, \eqref{eq:rate-condition-0} implies that there exists $\nu > 0$ such that 
\begin{align*}
\frac{1}{n} \log |\tilde{{\cal M}}_0^{(n)}| \le H(X) - \nu
\end{align*}
for sufficiently large $n$. Thus, by the continuity of the entropy, 
$\san{t}_{X^n Y^n} \in {\cal K}_n$ implies $\san{t}_{X^n} \in {\cal E}_n$.

Since $(r_0,r_2) \notin {\cal R}^*_{\mathtt{WAK}}(P_{XY})$, there exists $\mu \ge 0$ and $\nu > 0$ such that 
$r_0 + \mu r_2 \le R_\mu(P_{XY}) - (3+\mu)\nu$. 
Also, \eqref{eq:rate-condition-0} and \eqref{eq:rate-condition-2} imply 
\begin{align*}
\frac{1}{n} \log |\tilde{{\cal M}}_i^{(n)}| \le r_i + \nu,~~~i=0,2 
\end{align*}
for sufficiently large $n$. By Lemma \ref{lemma:continuity-1} and Lemma \ref{lemma:continuity-2}, 
$\san{t}_{X^nY^n} \in {\cal K}_n$ imply  
$R_\mu(P_{XY}) \le R_\mu(\delta_n | \san{t}_{X^nY^n}) + \nu$
for sufficiently large $n$, which implies 
\begin{align*}
\frac{1}{n} \log |\tilde{{\cal M}}_0^{(n)}| + \frac{\mu}{n} \log |\tilde{{\cal M}}_2^{(n)}|
&\le r_0 + \mu r_2 + (1+\mu)\nu \\
&\le R_\mu(P_{XY}) - 2\nu \\
&\le R_\mu(\delta_n | \san{t}_{X^n Y^n}) - \nu.
\end{align*}
Thus, $\san{t}_{X^n Y^n} \in {\cal K}_n$ implies 
$\tilde{r}_{0,n} + \mu \tilde{r}_{2,n} < R_\mu(\delta_n | \san{t}_{X^n Y^n})$, i.e.,
$(r_{0,n}, r_{2,n}) \notin {\cal R}^*_{\mathtt{WAK}}(\delta_n | \san{t}_{X^nY^n})$ for sufficiently large $n$. 

Consequently, by Corollary \ref{corollary:WAK-bound}, we have
\begin{align*}
\lefteqn{\rom{P}_{\mathtt{WAK}}(\tilde{\Phi}_n | P_{XY}^n) } \\
&\ge \rom{P}\bigg( (\tilde{r}_{0,n}, \tilde{r}_{2,n}) \notin {\cal R}^*_{\mathtt{WAK}}(\delta_n | \san{t}_{X^n Y^n}),~\san{t}_{X^n} \in {\cal E}_n \bigg) \bigg(1- \frac{1}{n} \bigg) \\
&\ge \rom{P}\bigg( \san{t}_{X^n Y^n} \in {\cal K}_n \bigg) \bigg(1- \frac{1}{n} \bigg) \\
&\ge \bigg( 1 - \frac{2 |{\cal X}||{\cal Y}|}{n^2} \bigg) \bigg(1- \frac{1}{n} \bigg),
\end{align*}
which implies \eqref{eq:convergence-to-one}. \qed

\section{Discussions} \label{section:discussion}

In this paper, in order to derive a converse bound on the WAK network from 
a converse bound on the GW network, we showed a reduction method to construct a GW code
from a given WAK code. Since the WAK network is distributed coding and the GW network is
centralized coding, an opposite reduction, i.e., constructing a WAK code from a given GW code,
is not possible in general. 

Since the residual terms in Corollary \ref{corollary:WAK-bound} are ${\cal O}((\log n)/n)$,
it may give an outer bound for the second-order region of the WAK network (cf.~\cite{watanabe:13e}).
However, $\delta=0$ could be singular points of the region ${\cal R}^*_{\mathtt{WAK}}(\delta |P_{XY})$
though this region is continuous at $\delta=0$. Thus, some careful treatment is needed to 
investigate a second-order outer bound, which is an interesting future research problem. 

Recently, a method to derive converse bounds for multiuser 
problems via reverse hypercontractivity was proposed by Liu-Handel-Verd\'u \cite{LiuHanVer:17};
they derived a second-order outer bound on the WAK network as an application of their approach.



\section*{Acknowlegement}

The author would like to thank Chandra Nair for letting the author know
the connection between the GW network and the WAK network. 
This work is supported
in part by JSPS KAKENHI Grant Number 16H06091.

\bibliographystyle{../../../09-04-17-bibtex/IEEEtranS}
\bibliography{../../../09-04-17-bibtex/reference.bib}
\end{document}